\newif\ifreport\reporttrue
\documentclass[conference,10pt,twocolumns]{IEEEtran}

\usepackage{graphicx}
\usepackage[center]{caption}
\usepackage{epsfig,latexsym,amsmath,amsfonts}
\usepackage{amssymb}
\usepackage{placeins}
\usepackage{subfigure}
\usepackage{verbatim}
\usepackage{cite,url}
\usepackage{epsf,bm}
\usepackage{color}
\usepackage{epstopdf}
\usepackage{romannum,hyperref}

\usepackage[utf8]{inputenc}
\usepackage[english]{babel}
\usepackage{amsthm}
\theoremstyle{definition}
\newtheorem{definition}{Definition}[section]
\usepackage[lined, boxed, linesnumbered, ruled]{algorithm2e}

\usepackage[utf8]{inputenc}
\usepackage[english]{babel}
\newtheorem{theorem}{Theorem}
\newtheorem{lemma}[theorem]{Lemma}

\begin{document}

\title{Age-Optimal Information Updates in Multihop Networks}

%\title{Minimizing the Age of Information in Multi-hop Information-Update Networks}
\author{\large Ahmed M. Bedewy$^\dagger$, Yin Sun$^\dagger$, and Ness B. Shroff$^\dagger$$^\ddagger$ \\ [.1in]
\large  \begin{tabular}{c} $^\dagger$Dept. of ECE, $^\ddagger$Dept. of CSE, The Ohio State University, Columbus, OH. \\
emails: \{bedewy.2, sun.745, shroff.11\}@osu.edu
%%email: bedewy.2@osu.edu, sun.745@osu.edu, shroff.11@osu.edu
\end{tabular} }
\maketitle

\begin{abstract}
The problem of reducing the age-of-information has been extensively studied in the single-hop networks. In this
paper, we minimize the age-of-information in general multihop networks. If the packet transmission times over the network links are exponentially distributed, we prove that a preemptive Last Generated First Served (LGFS) policy results in smaller age processes at all nodes of the network (in a stochastic ordering sense) than any other causal policy. In addition, for arbitrary general distributions of packet transmission times, the non-preemptive LGFS policy is shown to minimize the age processes at all nodes of the network among all  non-preemptive work-conserving policies (again in a stochastic ordering sense). It is surprising that such simple policies can achieve optimality of the joint distribution of the age processes at all nodes even under arbitrary network topologies, as well as arbitrary packet generation and arrival times. These optimality results not only hold for the age processes, but also for any non-decreasing functional of the age processes.

\end{abstract}
\section{Introduction}\label{Int}
There is a growing interest in applications that require real-time (fresh) information updates, such as news, weather reports, email notifications, stock quotes, social updates, mobile ads, etc. The freshness of the information is also crucial in other systems, e.g., monitoring systems that obtain information from environmental sensors, wireless systems that need rapid updates of channel state information, etc. 

To provide a precise metric of data freshness, the concept of \emph{age-of-information}, or simply \emph{age}, was defined in \cite{adelberg1995applying,cho2000synchronizing,golab2009scheduling,KaulYatesGruteser-Infocom2012}. At time $t$, if $U(t)$ is the time when the freshest update at the destination was generated, the age $\Delta(t)$  is defined as $\Delta(t)=t-U(t)$. Hence, age is the time elapsed since the freshest packet was generated.

There exists a number of studies that focused on reducing the age in a single-hop network \cite{KaulYatesGruteser-Infocom2012,2012ISIT-YatesKaul,2015ISITHuangModiano,CostaCodreanuEphremides2014ISIT,Icc2015Pappas,
2012CISS-KaulYatesGruteser,Gamma_dist, PAoI_in_error}. In \cite{KaulYatesGruteser-Infocom2012,2012ISIT-YatesKaul,2015ISITHuangModiano}, the update generation rate was optimized to improve data freshness in the class of First-Come First-Served (FCFS) update policies. In \cite{CostaCodreanuEphremides2014ISIT,Icc2015Pappas}, it was found that age can be improved by discarding old packets waiting in the queue if a new sample arrives. In \cite{2012CISS-KaulYatesGruteser,Gamma_dist}, the time-average age was characterized for Last-Come First-Served (LCFS) information-update networks with and without preemption with exponential and gamma service time distributions, respectively. The work in \cite{PAoI_in_error} complements the work in \cite{2012CISS-KaulYatesGruteser} by analyzing the average peak age in presence of error.

Age-optimal generation of update packets was studied for single-hop networks in \cite{BacinogCeranUysal_Biyikoglu2015ITA,2015ISITYates,generat_at_will}. In particular, a general class of non-negative, non-decreasing age penalty functions was minimized in \cite{generat_at_will}. 
%It was shown in \cite{2015ISITYates,generat_at_will} that a zero-wait policy, in which a new sample is taken once the previous sample is delivered, is throughput-optimal and delay-optimal but is not age-optimal. 
In \cite{age_optimality_multi_server}, it was shown that for arbitrary packet generation times, arrival times, and queue buffer size, a preemptive Last Generated First Served (LGFS) policy simultaneously minimizes the age, throughput, and delay in multi-server single-hop networks with exponential service times. Motivated by recent research on age-of-information, a real-time sampling problem was solved in  \cite{Sun_reportISIT17}, where samples of a Wiener process are taken and forwarded to a remote estimator via a channel with random delay. A simple sampling policy was developed to minimize the mean square estimation error subject to a sampling-rate constraint. 

% minimizing  that  receives samples casually from a sampler through a channel with random delay. 
%The optimal sampling strategy was proven to be a threshold policy, and the optimal threshold was found.
%It was observed that the behavior of the optimal update policy changes dramatically after adding a signal model. 

%Another important problem is how to maximize data freshness in information-update systems. This involves jointly controlling both the generation and transmission of packet updates .   In this setting, a  counter-intuitive phenomenon was revealed:  While a zero-wait  or work-conserving  policy, that generates and submits a fresh update once the server becomes idle, achieves the maximum throughput and the minimum average delay, surprisingly, this zero-wait policy does not always minimize the age. This implies that there is no policy that can simultaneously minimize age and maximize throughput, if the generation and transmission of update packets are jointly controlled. 

\begin{figure}
\includegraphics[scale=0.5]{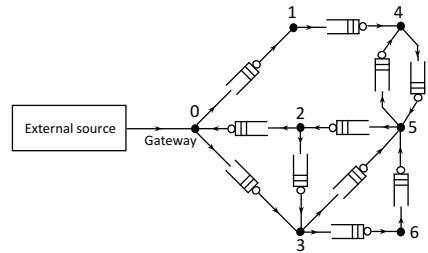}
\centering
\caption{ Information updates in a multihop network.}\label{Fig:sysMod}
\vspace{-0.3cm}
\end{figure}

In this paper, we consider a general multihop network, where the update packets are generated at an external source and are dispersed throughout the network through a gateway node, as shown in Fig. \ref{Fig:sysMod}. It is a step forward from \cite{KaulYatesGruteser-Infocom2012,2012ISIT-YatesKaul,2015ISITHuangModiano,CostaCodreanuEphremides2014ISIT,Icc2015Pappas,
2012CISS-KaulYatesGruteser,Gamma_dist, PAoI_in_error, BacinogCeranUysal_Biyikoglu2015ITA,2015ISITYates,generat_at_will, age_optimality_multi_server} by considering multihop networks. It is well known that delay-optimality is notoriously difficult in multihop networks, except for some special network settings (e.g., tandem networks) \cite{gupta2011delay,tassiulas1994dynamic}. This difficulty stems from the fact that packet scheduling decisions at each hop are influenced by decisions on other hops and vise versa. Surprisingly, it turns out that age minimization has very different features from delay minimization. In particular, we find that some simple policies can achieve optimality of the joint distribution of the age processes at all nodes even under arbitrary network topologies. The following summarizes our main contributions in this paper: 
\begin{itemize}
\item We consider a general scenario where the update packets do not necessarily  arrive to the gateway node in the order of their generation times. We prove that, if the packet transmission times over the network links are exponentially distributed, then for arbitrary arrival process, network topology, and buffer size at each link, the preemptive Last-Generated First-Served (LGFS) policy achieves smaller age processes at all nodes in the network  than any causal policy in the sense of stochastic ordering (Theorem \ref{thm1}). This further implies that the preemptive LGFS policy minimizes any non-decreasing functional of the age processes. Examples of non-decreasing age functional include the time-average age  \cite{KaulYatesGruteser-Infocom2012,2012ISIT-YatesKaul,CostaCodreanuEphremides2014ISIT,Icc2015Pappas,
2012CISS-KaulYatesGruteser,Gamma_dist,BacinogCeranUysal_Biyikoglu2015ITA,2015ISITYates}, average peak age \cite{2015ISITHuangModiano,CostaCodreanuEphremides2014ISIT,Gamma_dist,BacinogCeranUysal_Biyikoglu2015ITA,PAoI_in_error}, average age penalty \cite{generat_at_will}, etc. 
\item We then prove that, for arbitrary general distributions of packet transmission times, the non-preemptive LGFS policy minimizes the age processes at all nodes in the network among all non-preemptive work-conserving policies in the sense of stochastic ordering (Theorem \ref{thm2}). It is interesting to note that age-optimality here can be achieved even if the transmission time distribution differs from one link to another, i.e., the transmission time distributions are heterogeneous. 
%\item  Finally, we address the age performance of non-preemptive LGFS policy in a policy space that includes all non-preemptive policies. We consider the class of New-Better-than-Used (NBU) packet transmission time distributions and show that the non-preemptive LGFS policy is within three times of the optimum for minimizing the average age of all nodes in the network (Theorem \ref{thm3}).
\end{itemize}

\section{Model and Formulation}\label{sysmod}
\subsection{Notations and Definitions}
Throughout this paper, for any random variable $Z$ and an event $A$, let $[Z\vert A]$ denote a random variable with the conditional distribution of $Z$ for given $A$, and $\mathbb{E}[Z\vert A]$ denote the conditional expectation of $Z$ for given $A$.

Let $\mathbf{x}=(x_1,x_2,\ldots,x_n)$ and $\mathbf{y}=(y_1,y_2,\ldots,y_n)$ be two vectors in $\mathbb{R}^n$, then we denote $\mathbf{x}\leq\mathbf{y}$ if $x_i\leq y_i$ for $i=1,2,\ldots,n$. A set $U\subseteq \mathbb{R}^n$ is called upper if $\mathbf{y}\in U$ whenever $\mathbf{y}\geq\mathbf{x}$ and $\mathbf{x}\in U$. We will need the following definitions: 
\begin{definition} \textbf{ Univariate Stochastic Ordering:} \cite{shaked2007stochastic} Let $X$ and $Y$ be two random variables. Then, $X$ is said to be stochastically smaller than $Y$ (denoted as $X\leq_{\text{st}}Y$), if
\begin{equation*}
\begin{split}
\mathbb{P}\{X>x\}\leq \mathbb{P}\{Y>x\}, \quad \forall  x\in \mathbb{R}.
 \end{split}
\end{equation*}
\end{definition}
\begin{definition}\label{def_2} \textbf{Multivariate Stochastic Ordering:} \cite{shaked2007stochastic} 
Let $\mathbf{X}$ and $\mathbf{Y}$ be two random vectors. Then, $\mathbf{X}$ is said to be stochastically smaller than $\mathbf{Y}$ (denoted as $\mathbf{X}\leq_\text{st}\mathbf{Y}$), if
\begin{equation*}
\begin{split}
\mathbb{P}\{\mathbf{X}\in U\}\leq \mathbb{P}\{\mathbf{Y}\in U\}, \quad \text{for all upper sets} \quad U\subseteq \mathbb{R}^n.
 \end{split}
\end{equation*}
\end{definition}
\begin{definition} \textbf{ Stochastic Ordering of Stochastic Processes:} \cite{shaked2007stochastic} Let $\{X(t), t\in [0,\infty)\}$ and $\{Y(t), t\in[0,\infty)\}$ be two stochastic processes. Then, $\{X(t), t\in [0,\infty)\}$ is said to be stochastically smaller than $\{Y(t), t\in [0,\infty)\}$ (denoted by $\{X(t), t\in [0,\infty)\}\leq_\text{st}\{Y(t), t\in [0,\infty)\}$), if, for all choices of an integer $n$ and $t_1<t_2<\ldots<t_n$ in $[0,\infty)$, it holds that
\begin{align}\label{law9'}
\!\!\!(X(t_1),X(t_2),\ldots,X(t_n))\!\leq_\text{st}\!(Y(t_1),Y(t_2),\ldots,Y(t_n)),\!\!
\end{align}
where the multivariate stochastic ordering in \eqref{law9'} was defined in Definition \ref{def_2}.
\end{definition}

\subsection{Network Model}
We consider a general multihop network represented by a directed graph $\mathcal{G(\mathcal{V},\mathcal{L})}$ where $\mathcal{V}$ is the set of nodes and $\mathcal{L}$ is the set of links, as shown in Fig. \ref{Fig:sysMod}. The number of nodes in the network is $\vert\mathcal{V}\vert=N$. The update packets are generated at an external source, which is connected to the network through a gateway node $0$. The update packets are firstly forwarded to node 0, from which they are dispersed throughout the network. Let $(i, j)\in\mathcal{L}$ denote a link from node $i$ to node $j$, where $i$ is the origin node and $j$ is the destination node. Each link $(i,j)$ has a queue of buffer size $B_{ij}$ to store the incoming packets. The packet transmission time on each link $(i,j)$ is random.

%where each packet takes a time, called transmission time, to be sent from node $i$ to node $j$. 
%
%
%
%consisting of $N$ nodes. As shown in Fig. \ref{Fig:sysMod}, we represent the communication network by  
%
%At each link $(i, j)$, the packets are enqueued at a queue of buffer size $B_{ij}$; then, each packet takes a time, called transmission time, to be sent from node $i$ to node $j$.
%corresponds to a server with a queue of buffer size $B_{ij}$.
% The packet transmission times are assumed to be independent across links and \emph{i.i.d.} across time. Hence, there is no correlation across different links, which is satisfied when the network is interference-free, i.e., wired networks.
\subsection{Scheduling Policy}\label{Schpolicy}
 The system starts to operate at time $t=0$. A sequence of $n$ update packets are generated at the external source, where $n$ can be an arbitrary finite or infinite number. The generation time of the $i$-th packet is $s_i$, such that $0\leq s_1\leq s_2\leq \ldots\leq s_n$. We let $\pi$ denote a scheduling policy that determines when to send the packets on each link and in which order. Define $a_{ij}$ as the arrival time of the $i$-th packet to node $j$. Then, $s_i \leq a_{i0} \leq a_{ij}$ for all $j = 1, \ldots, N-1$. The packet generation times $(s_1, s_2, \ldots, s_n)$ and packet arrival times $(a_{10}, a_{20}, \ldots, a_{n0})$ at node $0$  are arbitrary given, which are independent of the scheduling policy. Note that the update packets may arrive at node $0$ out of the order of their generation times. For example, packet $i+1$ may arrive at node 0 earlier than packet $i$ such that $s_i \leq s_{i+1}$ but $a_{i0} \geq a_{(i+1)0}$. 
 
Let $\Pi$  denote the set of all {causal} policies, in which scheduling decisions are made based on the history and current state of the system. We define several types of policies in $\Pi$:

A policy is said to be \textbf{preemptive}, if a link can switch to send another packet at any time; the preempted packets will be stored back into the queue if there is enough buffer space and then sent out at a later time when the link is available again. In contrast, in a \textbf{non-preemptive} policy, a link must complete sending the current packet before starting to send another packet.

A policy is said to be \textbf{work-conserving}, if each link is kept fully utilized when there are packets
waiting in the queue feeding this link.
\subsection{Performance Metric}
\begin{figure}
\includegraphics[scale=0.22]{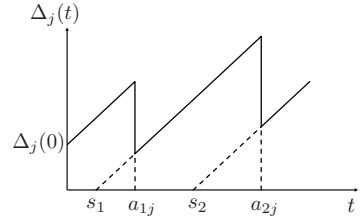}
\centering
\caption{Evolution of the age at node $\Delta_j(t)$ at node $j$.}\label{Fig:Age}
\vspace{-0.3cm}
\end{figure}
Let $U_{j}(t)=\max\{s_i : a_{ij}\leq t\}$ be the generation time of the freshest packet that has arrived at node $j$ before time $t$, where $U_j(0^-)$ is invariant of the policy $\pi\in\Pi$ for all $j\in\mathcal{V}$. The \emph{age-of-information}, or simply the \emph{age}, at node $j$ is defined as
\begin{equation}\label{age}
\begin{split}
\Delta_{j}(t)=t-U_{j}(t).
\end{split}
\end{equation} 
As shown in Fig. \ref{Fig:Age}, the age increases linearly with $t$ but is reset to a smaller value with the arrival of a fresher packet. The age vector of all the nodes in the network is $\mathbf{\Delta}(t)=\!\!(\Delta_{0}(t), \Delta_{1}(t), \ldots, \Delta_{N-1}(t))$. Then, the age process of all the nodes in the network is represented by $\mathbf{\Delta}=\{\mathbf{\Delta}(t), t\in [0,\infty)\}$. Let us define an age penalty functional $g(\mathbf{\Delta})$ to represent the level of ``dissatisfaction'' for data staleness in the network or the ``need'' for new information updates, where $g(\cdot)$ is a general non-decreasing functional: A functional $g$ is said to be non-decreasing, if 
\begin{equation}
\begin{split}
&g(\mathbf{\Delta}_1) \leq g(\bm{\Delta}_2),\\
&\text{whenever}\quad \Delta_{j,1}(t)\leq \Delta_{j,2}(t),\quad\!\! \forall t\in [0,\infty), \forall j\in\mathcal{V}. 
\end{split}
\end{equation}
Examples of the non-decreasing functional $g$ include:
\begin{itemize}
\item \emph{Time-average age \cite{KaulYatesGruteser-Infocom2012,2012ISIT-YatesKaul,CostaCodreanuEphremides2014ISIT,Icc2015Pappas,
2012CISS-KaulYatesGruteser,Gamma_dist,BacinogCeranUysal_Biyikoglu2015ITA,2015ISITYates}:} The time-average age of node $j$ is defined as
%\footnote{Assume that the limits in \eqref{functional1}, \eqref{functional2}, and \eqref{functional3} exist.}
\begin{equation}\label{functional1}
g_1(\bm{\Delta})=\frac{1}{T}\int_{0}^{T} \Delta_{j}(t) dt,
\end{equation}
\item \emph{Average peak age \cite{2015ISITHuangModiano,CostaCodreanuEphremides2014ISIT,Gamma_dist,BacinogCeranUysal_Biyikoglu2015ITA,PAoI_in_error}:} The average peak age of node $j$ is defined as 
\begin{equation}\label{functional2}
g_2(\bm{\Delta})=\frac{1}{K}\sum_{k=1}^{K} A_{kj},
\end{equation}
where $A_{kj}$ denotes the $k$-th peak value of $\Delta_{j}(t)$ since time $t=0$. 
\item \emph{Average age penalty \cite{generat_at_will}:} The average age penalty of node $j$ is
\begin{equation}\label{functional3}
g_3(\bm{\Delta})= \frac{1}{T}\int_{0}^{T} h(\Delta_{j}(t)) dt,
\end{equation}
where $h$ : $[0,\infty)\to [0,\infty)$ can be any non-negative and non-decreasing function.
\end{itemize}

\section{Age-optimality Results}\label{GS}
% \noindent \textbf{Definition 4.} \textbf{Age-Optimality:}  A policy $\beta \in \Pi$ is said to be \emph{age-optimal}, if for all $\pi\in\Pi$ 
%\begin{align}\label{law9}
%\{\Delta_{j,\beta}(t), t\in[0,\infty)\} 
% \leq_{\text{st}} \{\Delta_{j,\pi}(t), t\in[0,\infty)\},\quad\forall j,
%\end{align}
%where $\Delta_{j,\pi}(t)$ is the age at node $j$ under policy $\pi$ at time $t$. From the previous definition, we can deduce that the age-optimal policy minimizes the age at each node of the network. 
In this section, we first consider that the packet transmission time is exponentially distributed and show that the  preemptive Last Generated First Served (prmp-LGFS) policy is age optimal among all policies in the policy space $\Pi$. Then, we show that age-optimality can also be achieved for general transmission time distributions when the policy space is restricted to non-preemptive work-conserving policies. 
\subsection{Exponential Transmission Time Distributions}
 \begin{algorithm}[h]
\SetKwData{NULL}{NULL}
\SetCommentSty{small} 
$\alpha_{ij}:=0$\; 
\While{the system is ON} {
\If{a new packet with generation time $s$ arrives to node $i$}{ 
\uIf{the link is busy}{
\uIf{ $s\leq\alpha_{ij}$}
{Store the packet in the queue\;}
\Else(~~~~~~\tcp*[h]{The packet carries fresh information.}){
Send the packet over the link by preempting the packet being transmitted\; 
The preempted packet is stored back to the queue\;
 $\alpha_{ij}=s$\;
%\For {$i= 1, \ldots, m-1$}{
%\If{$\alpha_i< s\leq \alpha_{i+1}$}{$n=i$\;}}
%\If{$s>\alpha_m$}{$n=m$\;}
%$\alpha_n=s$\;
%%Sort ($\alpha_1, \ldots, \alpha_m,s$)\;
%%Assign the $m$ freshest packets to $\alpha_1, \ldots, \alpha_m$\;
%\For {$i= 1, \ldots, n-1$}{$\alpha_{i}:=\alpha_{i+1}$\;}
}}
\Else(~~~~~~\tcp*[h]{The link is idle.})
{
The new packet is sent over the link\;
} 

}
\If{a packet is delivered}{
 \If{ the queue is not empty}{
The freshest packet in the queue is sent over the link\;
 }
}
}
\caption{The preemptive Last Generated First Served (LGFS) policy.}\label{alg1}
\end{algorithm}
We study the age-optimal packet scheduling when the packet transmission times are exponentially distributed, \emph{independent} across the links and \emph{i.i.d.} across time. We consider a LGFS scheduling principle in which the packet being transmitted at each link is generated the latest (i.e., the freshest) one among all packets in the queue; after transmission, the link starts to send the next freshest packet in its queue. We propose a preemptive LGFS policy at each  link $(i,j)\in\mathcal{L}$. The implementation details of this policy are depicted in Algorithm \ref{alg1}. Throughout Algorithm \ref{alg1}, we use $\alpha_{ij}$ to denote the generation time of the packet being transmitted on the link $(i,j)$. 

%; after service, the server at each link starts to serve the next freshest packet stored in the queue feeding this link. 

 Define a set of parameters $\mathcal{I}=\{n,(s_i, a_{i0})_{i=1}^{n},$ $\mathcal{G}(\mathcal{V}, \mathcal{L}), (B_{ij}, (i,j)\in\mathcal{L})\}$, where $n$ is the total number of packets, $s_i$ and $a_{i0}$ are the generation time and the arrival time of packet $i$ to node $0$, respectively, $\mathcal{G}(\mathcal{V}, \mathcal{L})$ is the network graph, and $B_{ij}$ is the queue buffer size of link $(i,j)$. Let $\{\mathbf{\Delta_{\pi}}(t), t\in [0,\infty)\}$ be the age processes of all nodes in the network under policy $\pi$. The age performance of preemptive LGFS policy is provided in the following theorem.
\begin{theorem}\label{thm1}
If the packet transmission times are exponentially distributed, \emph{independent} across links and \emph{i.i.d.} across time, then for all $\mathcal{I}$ and $\pi\in\Pi$ 
\begin{align}
\!\!\!\![\mathbf{\Delta_{\text{prmp-LGFS}}}\vert\mathcal{I}]
 \!\!\leq_{\text{st}}\!\! [\mathbf{\Delta_{\pi}}\vert\mathcal{I}],\!\!\!
\end{align}
or equivalently, for all $\mathcal{I}$ and non-decreasing functional $g$
 \begin{equation}\label{thm1eq2}
\begin{split}
\mathbb{E}[g(\bm{\Delta}_{\text{prmp-LGFS}})\vert\mathcal{I}]= \min_{\pi\in\Pi} \mathbb{E}[g(\bm{\Delta}_\pi)\vert\mathcal{I}], 
\end{split}
\end{equation}
provided the expectations exist.
\end{theorem}
\begin{proof}
%We provide a proof sketch of Theorem \ref{thm1}.
% We use coupling and forward induction to prove it. 
 See Appendix~\ref{Appendix_A}.
 %We first consider the comparison between the preemptive LGFS policy and any work-conserving policy $\pi$. We couple the packets departure processes at each link of the network such that they are identical under both policies. Then, we use the forward induction over the packet arrival and departure events to show that the generation times of the freshest packets that have arrived at each node of the network are maximized under the preemptive LGFS policy. By this, the preemptive LGFS policy is age-optimal among the class of work-conserving policies. Finally, since the service times are independent across links and \emph{i.i.d.} across time, service idling only postpones the delivery of fresh packets at each node in the network. Therefore, the age under non-work-conserving policies will be greater. This completes the proof. 
\end{proof}

Theorem \ref{thm1} tells us that for arbitrary number $n$, packet generation times $(s_1, s_2, \ldots, s_n)$ and arrival times $(a_{10}, a_{20}, \ldots, a_{n0})$ at node $0$, network topology $\mathcal{G(V, L)}$, and buffer sizes $(B_{ij},(i,j)\in\mathcal{L})$, the prmp-LGFS policy can achieve optimality of the joint distribution of the age processes at all nodes in the network within the policy space $\Pi$. In addition, \eqref{thm1eq2} tells us that the prmp-LGFS policy minimizes any non-decreasing age penalty functional $g$, including the time-average age \eqref{functional1}, average peak age \eqref{functional2}, and average age penalty \eqref{functional3}.

\subsection{General Transmission Time Distributions}
Now, we study the age-optimal packet scheduling for \emph{arbitrary} general packet transmission time distributions which are \emph{independent} across the links and \emph{i.i.d.} across time. We consider the set of non-preemptive work-conserving policies, denoted by $\Pi_{npwc}\subset\Pi$. We propose a non-preemptive LGFS policy. The description of non-preemptive LGFS policy can be obtained from Algorithm \ref{alg1} by replacing Steps 5-11 by Step 6. We next show that the non-preemptive LGFS policy is age-optimal among the policies in $\Pi_{npwc}$.
\begin{theorem}\label{thm2}
If the packet transmission times are \emph{independent} across the links and \emph{i.i.d.} across time, then for all $\mathcal{I}$ and $\pi\in\Pi_{npwc}$ 
\begin{align}
\!\!\!\![\mathbf{\Delta_{\text{non-prmp-LGFS}}}\vert\mathcal{I}]
 \!\!\leq_{\text{st}}\!\! [\mathbf{\Delta_{\pi}}\vert\mathcal{I}],\!\!\!
\end{align}
or equivalently, for all $\mathcal{I}$ and non-decreasing functional $g$
 \begin{equation}
\begin{split}
\mathbb{E}[g(\bm{\Delta}_{\text{non-prmp-LGFS}})\vert\mathcal{I}]= \min_{\pi\in\Pi_{npwc}} \mathbb{E}[g(\bm{\Delta}_\pi)\vert\mathcal{I}], 
\end{split}
\end{equation}
provided the expectations exist.
\end{theorem}
\begin{proof}
 \ifreport
See Appendix~\ref{Appendix_B}.
\else
The proof of Theorem \ref{thm2} is similar to that of Theorem \ref{thm1}. The difference is that the preemption is not allowed here. See our technical report for more details \cite{Technical_report}.\fi
\end{proof}

It is interesting to note from Theorem \ref{thm2} that, age-optimality can be achieved for arbitrary general transmission time distributions, even if the transmission time distribution differs from a link to another.

\section{Numerical Results}\label{Simulations}
\begin{figure}
\includegraphics[scale=0.5]{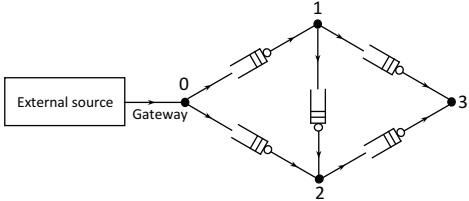}
\centering
\caption{ A  multihop network.}\label{Fig:simulation}
\vspace{-0.3cm}
\end{figure}
We present some numerical results to illustrate the age performance of different policies and validate the theoretical results. We consider the network in Fig. \ref{Fig:simulation}. The inter-generation times are \emph{i.i.d.} Erlang-2 distribution with mean $1/\lambda$. The time difference between packet generation and arrival to node $0$, i.e., $a_{i0}-s_i$, is modeled to be either  1 or 100, with equal probability. This means that the update packets may arrive to node $0$ out of order of their generation time. 

%Next, we provide the simulation results that match with Theorems \ref{thm1} and \ref{thm2}. The simulation result for Theorem \ref{thm3} is provided in our technical report.

%We consider the following reference policies for comparison: The FCFS policy with infinite buffer and buffer size B=1 was introduced in , 

%The packet service times are exponentially distributed with mean $1/\mu =1$, which is \emph{i.i.d.} across time and servers.
 
\begin{figure}[t]
\centering
\includegraphics[scale=0.3]{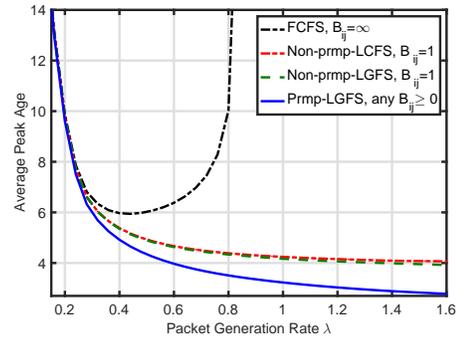}
\caption{Average peak age at node 2 versus packets generation rate $\lambda$ for exponential packet transmission times.}
\vspace{-.0in}
\label{avg_age1}
\end{figure}
Figure \ref{avg_age1} illustrates the average peak age at node 2 versus the packet generation rate $\lambda$ for the multihop network in Fig. \ref{Fig:simulation}. The packet transmission times are exponentially distributed with mean 1 at links $(0,1)$ and $(1,2)$, and mean 0.5 at link $(0,2)$. One can observe that the preemptive LGFS policy achieves a better (smaller) peak age at node $2$ than the non-preemptive LGFS policy, non-preemptive LCFS policy, and FCFS policy, where the buffer sizes are either 1 or infinity. It is important to emphasize that the peak age is minimized by preemptive LGFS policy for out of order packet receptions at node $0$, and general network topology. This numerical result shows agreement with Theorem \ref{thm1}.

%the FCFS policy  analyzed in \cite{KaulYatesGruteser-Infocom2012}, and the non-preemptive LCFS policy with queue size $B=1$ \cite{2012CISS-KaulYatesGruteser} which was also named ``M/M/1/2*'' in \cite{CostaCodreanuEphremides2014ISIT}. 
%Note that  in these prior studies, the time-average age was  characterized only for the special case of Poisson arrival process.   Moreover, with ordered arrived packets at the server, the LGFS policy and LCFS policy have the same age performance.

%To validate the generality of our theoretical results, the arrival process considered here is different from the Poisson arrival process analyzed in \cite{KaulYatesGruteser-Infocom2012,2012ISIT-YatesKaul,2015ISITHuangModiano,CostaCodreanuEphremides2014ISIT}: 
%compared to the policies proposed in prior studies. 
\begin{figure}[t]
\centering
\includegraphics[scale=0.3]{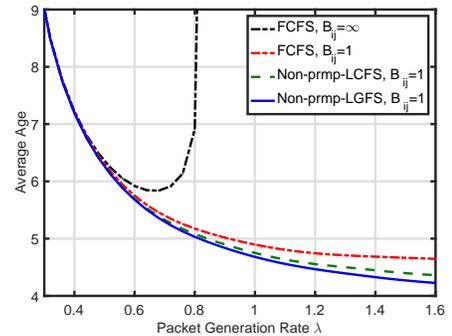}
\caption{Average age at node 3 versus packets generation rate $\lambda$ for general packet transmission time distributions.}
\vspace{-.0in}
\label{avg_age2}
\end{figure}

Figure \ref{avg_age2} plots  the time-average age at node 3 versus the packets generation rate $\lambda$ for the multihop network in Fig. \ref{Fig:simulation}. The plotted policies are FCFS policy, non-preemptive LCFS, and non-preemptive LGFS policy, where the buffer sizes are either 1 or infinity. The packet transmission times at links $(0,1)$ and $(1,3)$ follow a gamma distribution with mean 1. The packet transmission times at links $(0, 2)$, $(1, 2)$, and $(2, 3)$ are distributed as the sum of a constant with value 0.5 and a value drawn from an exponential distribution with mean 0.5. We find that the non-preemptive LGFS policy achieves the best age performance among all plotted policies. By comparing the age performance of the non-preemptive LGFS  and non-preemptive LCFS policies, we observe that the LGFS scheduling principle improves the age performance when the update packets arrive to node 0 out of the order of their generation times. It is important to note that the non-preemptive LGFS policy minimizes the age among the non-preemptive work-conserving policies even if the packet transmission time distributions are heterogeneous across the links. We also observe that the average age of FCFS policy with $B_{ij}=\infty$ blows up when the traffic intensity is high. This is due to the increased congestion in the network which leads to a delivery of stale packets. Moreover, in case of the FCFS policy with $B_{ij}=1$, the average age is finite at high traffic intensity, since the fresh packet has a better opportunity to be delivered in a relatively short period compared with FCFS policy with $B_{ij}=\infty$. This numerical result agrees with Theorem \ref{thm2}.

\section{Conclusion}\label{Concl}
In this paper, we made the first attempt to minimize the age-of-information in  general multihop networks. We showed that for general system settings including arbitrary network topology, packet generation times and arrival times to node 0, and queue buffer sizes, the age-optimality can be achieved. These optimality results not only hold for the age processes, but also for any non-decreasing functional of the age processes.

%It was showed that, if the packet transmission times are exponentially distributed, then for any given arrival process at node $0$, any given network topology and any given buffer sizes, the preemptive LGFS policy is age-optimal among all causally feasible policies. Also, among the non-preemptive work-conserving policies, we showed that the non-preemptive LGFS policy is age-optimal when the packet transmission time distributions are general. 

%Finally, we show that among all non-preemptive policies, the non-preemptive LGFS scheme is within three times the optimum for minimizing the average age of all nodes in the network when the packet transmission times are NBU. 

%In this paper, we considered an information-update system, in which update packets are forwarded to a destination through multiple network servers. It was showed that, if the packet service times are \emph{i.i.d.} exponentially distributed, then for any given arrival process and queue size, the preemptive LGFS policy simultaneously optimizes the data freshness, throughput, and delay performance among all causally feasible policies. We will extend these results to more general system settings with general service time distributions.   
\appendices
\section{Proof of Theorem \ref{thm1}}\label{Appendix_A}
We need to define the system state of any policy $\pi$:

 \definition  At any time $t$, the system state of policy $\pi$ is specified by  $\mathbf{U}_\pi(t)=( U_{0,\pi}(t), U_{2,\pi}(t),\ldots,U_{N-1,\pi}(t)) $, where $U_{j,\pi}(t)$ is the generation time of the freshest packet that have already arrived to node $j$ by time $t$.
Let $\{\mathbf{U}_\pi(t), t\in[0,\infty)\}$ be the state process of policy $\pi$, which is assumed to be right-continuous. For notational simplicity, let policy $P$ represent the preemptive LGFS policy.

The key step in the proof of Theorem \ref{thm1} is the following lemma, where we compare policy $P$ with any work-conserving policy $\pi$.

 \begin{lemma}\label{lem2}
 Suppose that $\mathbf{U}_{P}(0^-)=\mathbf{U}_{\pi}(0^-)$ for all work conserving policies $\pi$, then for all $\mathcal{I}$,
\begin{equation}\label{law9}
\begin{split}
[\{\mathbf{U}_{P}(t),  t\in[0,\infty)\}\vert\mathcal{I}]\geq_{\text{st}}[\{\mathbf{U}_{\pi}(t), t\in[0,\infty)\}\vert\mathcal{I}].
 \end{split}
\end{equation}
\end{lemma}

 %It is worth noting that, all system state processes have the same arrival time sequence and the same service time distribution. 
 We use coupling and forward induction to prove Lemma \ref{lem2}.
For any work-conserving policy $\pi$, suppose that stochastic processes $\widetilde{\mathbf{U}}_{P}(t)$ and $\widetilde{\mathbf{U}}_{\pi}(t)$ have the same distributions with $\mathbf{U}_{P}(t)$ and $\mathbf{U}_{\pi}(t)$, respectively. 
%Let $\hat{V}_{P}(t)=(\hat{U}_{P}(t),\hat{ \alpha}_{1,P}(t),\ldots,\hat{\alpha}_{m,P}(t))$ and $\{\hat{V}_{\pi}(t),  t\in[0,\infty)\}=\{\hat{U}_{\pi}(t),\hat{ \alpha}_{1,\pi}(t),\ldots,\hat{\alpha}_{m,\pi}(t), t\in[0,\infty)\}$. 
The state processes $\widetilde{\mathbf{U}}_{{P}}(t)$ and $\widetilde{\mathbf{U}}_{\pi}(t)$ are coupled in the following manner: If a packet is delivered from node $i$ to node $j$ through the link $(i,j)$ at time $t$ as $\widetilde{\mathbf{U}}_{{P}}(t)$ evolves in policy $P$,  then there exists a packet delivery from node $i$ to node $j$ through the link $(i,j)$ at time $t$ as $\widetilde{\mathbf{U}}_{\pi}(t)$ evolves in policy $\pi$. 
Such a coupling is valid since the transmission time is exponentially distributed and thus memoryless. Moreover, policy ${P}$ and policy $\pi$ have identical packet generation times $(s_1, s_2, \ldots, s_n)$ and packet arrival times $(a_{10}, a_{20}, \ldots, a_{n0})$ to node 0. According to Theorem 6.B.30 in \cite{shaked2007stochastic}, if we can show 
\begin{equation}\label{main_eq}
\begin{split}
\mathbb{P}[\widetilde{\mathbf{U}}_{P}(t)\geq\widetilde{\mathbf{U}}_{\pi}(t), t\in[0,\infty)\vert\mathcal{I}]=1,
\end{split}
\end{equation}
then \eqref{law9} is proven. 

To ease the notational burden, we will omit the tildes in this proof on the coupled versions and just use $\mathbf{U}_{P}(t)$ and $\mathbf{U}_{\pi}(t)$. Next, we use the following lemmas to prove \eqref{main_eq}:

\begin{lemma}\label{lem3}
Suppose that under policy $P$, $\mathbf{U'}_{P}$ is obtained by a packet delivery at the link $(i,j)$ in the system whose state is $\mathbf{U}_{P}$. Further, suppose that under policy $\pi$, $\mathbf{U'}_{\pi}$ is obtained by a packet delivery at the link $(i,j)$ in the system whose state is $\mathbf{U}_\pi$. If
\begin{equation}\label{hyp1}
 \mathbf{U}_{P} \geq \mathbf{U}_\pi,
\end{equation}
then,
\begin{equation}\label{law6}
\mathbf{U'}_{P} \geq \mathbf{U'}_{\pi}.
\end{equation}
\end{lemma}

\begin{proof}\ifreport
Let $s_{P}$ and $s_\pi$ denote the generation times of the packets that are delivered over the link $(i,j)$ under policy $P$ and policy $\pi$, respectively. From the definition of the system state, we can deduce that
\begin{equation}\label{Def1}
\begin{split}
U_{j,P}'&=\max\{U_{j,P},s_{P}\},\\
U_{j,\pi}'&=\max\{U_{j,\pi},s_{\pi}\}.
\end{split}
\end{equation}
Hence, we have two cases:

Case 1: If $s_{P}\geq s_\pi$. From \eqref{hyp1}, we have
\begin{equation}\label{pfp21}
U_{j,P}\geq U_{j,\pi}.
\end{equation}
Also, $s_{P}\geq s_\pi$, together with  \eqref{Def1} and \eqref{pfp21} imply
\begin{equation}
U'_{j,P}\geq U'_{j,\pi}.
\end{equation}
Since there is no packet delivery under other links, we get
\begin{equation}
U'_{k,P}=U_{k,P}\geq U_{k,\pi}=U'_{k,\pi}, \quad \forall k\neq j.
\end{equation}
Hence, we have 
\begin{equation}
\mathbf{U'}_{P} \geq \mathbf{U'}_{\pi}.
\end{equation}

Case 2: If $s_{P}<s_\pi$. By the definition of the system state, $s_{P} \leq U_{i,P}$ and $s_\pi \leq U_{i,\pi}$. Then, using $U_{i, P}\geq U_{i,\pi}$, we obtain
\begin{equation}
s_{P} < s_\pi \leq U_{i,\pi} \leq U_{i, P}.
\end{equation}
Because $s_{P}<U_{i, P}$, policy $P$ is sending a stale packet on link $(i,j)$. By the definition of policy $P$, this happens only when all packets generated after $s_P$ in the queue of the link $(i,j)$ have been delivered to node $j$. Since $s_\pi \leq U_{i, P}$, node $i$ has already received a packet (say packet $w$) generated no earlier than $s_\pi$ in policy $P$. Because $s_{P}<s_\pi$, packet $w$ is generated after $s_{P}$. Hence, packet $w$ must have been delivered to node $j$ in policy $P$ such that
\begin{equation}\label{eqprmp1}
 s_\pi  \leq U_{j, P}.
\end{equation}
Also, from \eqref{hyp1}, we have
\begin{equation}\label{eqprmp2}
 U_{j,\pi} \leq U_{j, P}.
\end{equation}
Combining \eqref{eqprmp1} and \eqref{eqprmp2} with \eqref{Def1}, we obtain
\begin{equation}
U'_{j,P}\geq U'_{j,\pi}.
\end{equation}
Since there is no packet delivery under other links, we get
\begin{equation}
U'_{k,P}=U_{k,P}\geq U_{k,\pi}=U'_{k,\pi}, \quad \forall k\neq j.
\end{equation}
Hence, we have 
\begin{equation}
\mathbf{U'}_{P} \geq \mathbf{U'}_{\pi},
\end{equation}
%
%
%
%Since $U_{k,P_{\text{LGFS}}}\geq U_{k,\pi}$, the packet $B$ or a fresher one must have arrived to node $k$ under policy $P_{\text{LGFS}}$ before the considered time epoch. We claim that packet $B$ must have been delivered from node $k$ to node $l$ under policy $P_{\text{LGFS}}$ before the considered time epoch. To prove this claim, let us assume, for the sake of contradiction, that packet $B$ has not been delivered from node $k$ to node $l$ under policy $P_{\text{LGFS}}$ before the considered time epoch. Since policy $P_{\text{LGFS}}$ is preemptive Last Generated First Served policy and $s_{P_{\text{LGFS}}}<s_\pi$, packet $B$ must have preempted the service of the packet $A$ under policy $P_{\text{LGFS}}$. This contradicts the fact that packet $A$ is served at link $(k,l)$. Hence, our claim is proven. As a result we have
%\begin{equation}
%U_{l,P_{\text{LGFS}}}\geq s_\pi.
%\end{equation}
%Combining this with \eqref{hyp1} and \eqref{Def1}, we obtain
%\begin{equation}
%U'_{l,P_{\text{LGFS}}}\geq U'_{l,\pi}.
%\end{equation}
%Since there is no packet delivery under other links, we get
%\begin{equation}
%U'_{i,P_{\text{LGFS}}}=U_{i,P_{\text{LGFS}}}\geq U_{i,\pi}=U'_{i,\pi}, \quad \forall i\neq l.
%\end{equation}
%Hence, we have 
%\begin{equation}
%U_{i,P_{\text{LGFS}}}' \geq U_{i,\pi}', \quad \forall i=1,\ldots,N-1.
%\end{equation}
which complete the proof. 
\else
See our technical report \cite{Technical_report}.\fi
\end{proof}

\begin{lemma}\label{lem4}
Suppose that under policy $P$, $\mathbf{U'}_{P}$ is obtained by the arrival of a new packet to node $0$ in the system whose state is $\mathbf{U}_{P}$. Further, suppose that under policy $\pi$, $\mathbf{U'}_{\pi}$ is obtained by the arrival of a new packet to node $0$ in the system whose state is $\mathbf{U}_\pi$. If
\begin{equation}\label{hyp2}
 \mathbf{U}_{P} \geq \mathbf{U}_\pi,
\end{equation}
then,
\begin{equation}
\mathbf{U'}_{P} \geq \mathbf{U'}_{\pi}.
\end{equation}
\end{lemma}

\begin{proof}\ifreport
Let $s$ denote the generation time of the new arrived packet. From the definition of the system state, we can deduce that
\begin{equation}\label{Def2}
\begin{split}
U_{0,P}'&=\max\{U_{0,P},s\},\\
U_{0,\pi}'&=\max\{U_{0,\pi},s\}.
\end{split}
\end{equation}
Combining this with \eqref{hyp2}, we obtain
\begin{equation}
U'_{0,P}\geq U'_{0,\pi}.
\end{equation}
Since there is no packet delivery under other links, we get
\begin{equation}
U'_{k,P}=U_{k,P}\geq U_{k,\pi}=U'_{k,\pi}, \quad \forall k\neq 0.
\end{equation}
Hence, we have 
\begin{equation}
\mathbf{U'}_{P} \geq \mathbf{U'}_{\pi},
\end{equation}
which complete the proof.
\else
See our technical report \cite{Technical_report}.\fi 
\end{proof}

\begin{proof}[ Proof of Lemma \ref{lem2}]
For any sample path, we have that $\mathbf{U}_{P}(0^-) = \mathbf{U}_{\pi}(0^-)$. This, together with Lemma \ref{lem3} and Lemma \ref{lem4},  implies that  
\begin{equation}
\begin{split}
[\mathbf{U}_{P}(t)\vert\mathcal{I}] \geq [\mathbf{U}_{\pi}(t)\vert\mathcal{I}],\nonumber
\end{split}
\end{equation}
holds for all $t\in[0,\infty)$. Hence, \eqref{main_eq} holds which implies \eqref{law9} by Theorem 6.B.30 in \cite{shaked2007stochastic}.
%\begin{equation*}\label{law9}
%\begin{split}
% \{V_{P}(t),  t\in[0,\infty)\}\!\geq_{\text{st}}\! \{V_{\pi}(t), t\in[0,\infty)\} \quad \forall \pi\in\Pi_{wc},\!\!\!\!\!\!
% \end{split}
%\end{equation*}
%As a result, we get 
%\begin{equation*}\label{law9}
%\begin{split}
%&\{U_{p}(t),  t\in[0,\infty)\}\geq_{\text{st}} \{U_{\pi}(t), t\in[0,\infty)\},\\
%&\{\alpha_{i,p}(t),  t\in[0,\infty)\}\!\!\geq_{\text{st}}\!\! \{\alpha_{i,\pi}(t), t\in[0,\infty)\}, i\!\!=\!\!1,\ldots,m,\!\!\!\!
% \end{split}
%\end{equation*}
%holds for all $\pi\in\Pi_{wc}$, 
This completes the proof.
\end{proof}

% In the following proof, we assumed that different policies can have exact the same packet delivery time instants and both packets have the same time stamp order. This assumption is valid due to the memory-less property of the exponential service time and the arbitrariness of the arrival process. 
%
%First, we will show that (\ref{main_eq}) holds during $[s_1,s_2)$. Because $U_p (0^-) =U_{\pi}(0^-)$ and $ \alpha_{i,p}(0^-) = \alpha_{i,\pi}(0^-)$ for $i=1,\ldots,m$,   (\ref{main_eq}) holds trivially at $t=0^-$. Update 1 arrives at time $s_1 =0$. By Lemma \ref{lem3}, (\ref{main_eq}) holds at time $t=0$. If there is a departure in $[s_1,s_2)$, then (\ref{main_eq}) holds by Lemma \ref{lem4}. If there is no departure, then (\ref{main_eq}) holds trivially. Therefor, (\ref{main_eq})  holds for all $t\in (s_1, s_2)$.
%
%Next, we prove the induction step. Suppose that (\ref{main_eq}) holds at time $s^-$, we will show that (\ref{main_eq}) holds
%during $[s_j,s_{j+1})$. Update j arrives at time $s_j$. By Lemma \ref{lem3}, (\ref{main_eq}) holds at time $s_j$. If there are one or more departures during $[s_j,s_{j+1})$, then by Lemma \ref{lem4}, (\ref{main_eq}) holds for all $t\in (s_j,s_{j+1})$. If there is no departure, then (\ref{main_eq}) holds trivially for all $t\in (s_j,s_{j+1})$. Thus Lemma \ref{lem2} is proven by using induction.

\begin{proof}[Proof of Theorem \ref{thm1}]
As a result of Lemma \ref{lem2}, we have
\begin{equation*}
\begin{split}
[\{\mathbf{U}_{P}(t),  t\in[0,\infty)\}\vert\mathcal{I}]\geq_{\text{st}} [\{\mathbf{U}_{\pi}(t), t\in[0,\infty)\}\vert\mathcal{I}],
 \end{split}
\end{equation*}
holds for all work-conserving policies $\pi$, which implies
\begin{equation*}
\begin{split}
[\{\mathbf{\Delta}_{P}(t), t\in[0,\infty)\}\vert\mathcal{I}]\!\!\leq_{\text{st}} \!\![\{\mathbf{\Delta}_{\pi}(t), t\in[0,\infty)\}\vert\mathcal{I}],
 \end{split}
\end{equation*}
holds for all work-conserving policies $\pi$.

Finally, transmission idling only postpones the delivery of fresh packets. Therefore, the age under non-work-conserving policies will be greater. As a result, we have 
\begin{equation*}
\begin{split}
[\{\mathbf{\Delta}_{P}(t), t\in[0,\infty)\}\vert\mathcal{I}]\!\!\leq_{\text{st}} \!\![\{\mathbf{\Delta}_{\pi}(t), t\in[0,\infty)\}\vert\mathcal{I}],
 \end{split}
\end{equation*}
holds for all $\pi\in\Pi$. This completes the proof.
%Note that this proof argument holds for any buffer size $B\geq0$. 
%which completes the proof. 
\end{proof}

\ifreport
\section{Proof of Theorem \ref{thm2}}\label{Appendix_B}
This proof is similar to that of Theorem \ref{thm1}. The difference between this proof and the proof of Theorem \ref{thm1} is that policy $\pi$ cannot be a preemptive policy here. We will use the same definition of the system state of policy $\pi$ used in Theorem \ref{thm1}. For notational simplicity, let policy $NP$ represent the non-preemptive LGFS policy.

The key step in the proof of Theorem \ref{thm2} is the following lemma, where we compare policy $NP$ with an arbitrary policy $\pi\in\Pi_{npwc}$.
% \begin{lemma}\label{lem2np}
%For any given multihop network topology, any given packet generation times at the external source and arrival times to node $0$, and any given buffer size $B_{ij}$ at each link $(i,j)$, suppose that $V_{NP_{\text{LGFS}}}(0^-)=V_{\pi}(0^-)$ for all policies $\pi\in \Pi'_{wc}$, then
%\begin{equation}\label{law9np}
%\begin{split}
% \{V_{NP_{\text{LGFS}}}(t),  t\in[0,\infty)\}\!\geq_{\text{st}}\! \{V_{\pi}(t), t\in[0,\infty)\}.
% \end{split}
%\end{equation}
%\end{lemma}
 \begin{lemma}\label{lem2np}
 Suppose that $\mathbf{U}_{NP}(0^-)=\mathbf{U}_{\pi}(0^-)$ for all $\pi\in\Pi_{npwc}$, then for all $\mathcal{I}$,
\begin{equation}\label{law9np}
\begin{split}
[\{\mathbf{U}_{NP}(t),  t\in[0,\infty)\}\vert\mathcal{I}]\!\geq_{\text{st}}\! [\{\mathbf{U}_{\pi}(t), t\in[0,\infty)\}\vert\mathcal{I}].
 \end{split}
\end{equation}
\end{lemma}

% %It is worth noting that, all system state processes have the same arrival time sequence and the same service time distribution. 
% We use coupling and forward induction to prove Lemma \ref{lem2np}.
%For any work-conserving policy $\pi$, suppose that stochastic processes $\widetilde{V}_{NP_{\text{LGFS}}}(t)$ and $\widetilde{V}_{\pi}(t)$ have the same stochastic laws as $V_{NP_{\text{LGFS}}}(t)$  and $V_{\pi}(t)$. 
%%Let $\hat{V}_{P}(t)=(\hat{U}_{P}(t),\hat{ \alpha}_{1,P}(t),\ldots,\hat{\alpha}_{m,P}(t))$ and $\{\hat{V}_{\pi}(t),  t\in[0,\infty)\}=\{\hat{U}_{\pi}(t),\hat{ \alpha}_{1,\pi}(t),\ldots,\hat{\alpha}_{m,\pi}(t), t\in[0,\infty)\}$. 
%The packet deliveries at each link of the network of the state processes $\widetilde{V}_{NP_{\text{LGFS}}}(t)$ and $\widetilde{V}_{\pi}(t)$ are coupled in the following manner: For each link $(i,j)$ in the network, If there exist a packet delivery at this link $(i,j)$ (i.e., a packet is delivered from node $i$ to node $j$) at time $t$ as $\widetilde{V}_{NP_{\text{LGFS}}}(t)$ evolves, then there exist a packet delivery at the same link $(i,j)$ at time $t$  as $\widetilde{V}_{\pi}(t)$ evolves. 
 We use coupling and forward induction to prove Lemma \ref{lem2np}.
For any work-conserving policy $\pi$, suppose that stochastic processes $\widetilde{\mathbf{U}}_{NP}(t)$ and $\widetilde{\mathbf{U}}_{\pi}(t)$ have the same distributions with $\mathbf{U}_{NP}(t)$ and $\mathbf{U}_{\pi}(t)$, respectively. 
%Let $\hat{V}_{P}(t)=(\hat{U}_{P}(t),\hat{ \alpha}_{1,P}(t),\ldots,\hat{\alpha}_{m,P}(t))$ and $\{\hat{V}_{\pi}(t),  t\in[0,\infty)\}=\{\hat{U}_{\pi}(t),\hat{ \alpha}_{1,\pi}(t),\ldots,\hat{\alpha}_{m,\pi}(t), t\in[0,\infty)\}$. 
The state processes $\widetilde{\mathbf{U}}_{{NP}}(t)$ and $\widetilde{\mathbf{U}}_{\pi}(t)$ are coupled in the following manner: If a packet is delivered from node $i$ to node $j$ through the link $(i,j)$ at time $t$ as $\widetilde{\mathbf{U}}_{{NP}}(t)$ evolves in policy prmp-LGFS,  then there exists a packet delivery from node $i$ to node $j$ through the link $(i,j)$ at time $t$ as $\widetilde{\mathbf{U}}_{\pi}(t)$ evolves in policy $\pi$.
Such a coupling is valid since the transmission time distribution at each link is identical under all policies. Moreover, policy $\pi$ can not be either preemptive or non-work-conserving policy, and both policies have the same packets generation times $(s_1, s_2, \ldots, s_n)$ and arrival times $(a_{10}, a_{20}, \ldots, a_{n0})$ to node 0. According to Theorem 6.B.30 in \cite{shaked2007stochastic}, if we can show 
\begin{equation}\label{main_eqnp}
\begin{split}
\mathbb{P}[\widetilde{\mathbf{U}}_{NP}(t)\geq\widetilde{\mathbf{U}}_{\pi}(t), t\in[0,\infty)\vert\mathcal{I}]=1,
\end{split}
\end{equation}
then \eqref{law9np} is proven.  

To ease the notational burden, we will omit the tildes henceforth on the coupled versions and just use $\mathbf{U}_{NP}(t)$ and $\mathbf{U}_{\pi}(t)$.

Next, we use the following lemmas to prove \eqref{main_eqnp}:

\begin{lemma}\label{lem3np}
Suppose that under policy $NP$, $\mathbf{U}_{NP}(\nu)$ is obtained by a packet delivery at the link $(i,j)$ at time $\nu$ in the system whose state is $\mathbf{U}_{NP}(\nu^{-})$. Further, suppose that under policy $\pi$, $\mathbf{U}_{\pi}(\nu)$ is obtained by a packet delivery at the link $(i,j)$ at time $\nu$ in the system whose state is $\mathbf{U}_\pi(\nu^{-})$. If
\begin{equation}\label{hyp1np}
\begin{split}
\mathbf{U}_{NP}(t)& \geq \mathbf{U}_\pi(t),
\end{split}
\end{equation}
holds for all $t\in [0, \nu^{-}]$, then
\begin{equation}\label{law6np}
\mathbf{U}_{NP}(\nu) \geq \mathbf{U}_{\pi}(\nu).
\end{equation}
\end{lemma}

\begin{proof}
Let $s_{NP}$ and $s_\pi$ denote the packet indexes and the generation times of the delivered packets over the link $(i,j)$ at time $\nu$ under policy $NP$ and policy $\pi$, respectively.  From the definition of the system state, we can deduce that
\begin{equation}\label{Def1np}
\begin{split}
U_{j,NP}(\nu)&=\max\{U_{j,NP}(\nu^{-}),s_{NP}\},\\
U_{j,\pi}(\nu)&=\max\{U_{j,\pi}(\nu^{-}),s_{\pi}\}.
\end{split}
\end{equation}
Hence, we have two cases:

Case 1: If $s_{NP}\geq s_\pi$. From \eqref{hyp1np}, we have
\begin{equation}\label{pf21}
U_{j,NP}(\nu^{-})\geq U_{j,\pi}(\nu^{-}).
\end{equation}
By $s_{NP}\geq s_\pi$, \eqref{Def1np}, and \eqref{pf21}, we have
\begin{equation}
U_{j,NP}(\nu)\geq U_{j,\pi}(\nu).
\end{equation}
Since there is no packet delivery under other links, we get
\begin{equation}
\begin{split}
\!\!\!\!U_{k,NP}(\nu)&=U_{k,NP}(\nu^{-})\\&\geq U_{k,\pi}(\nu^{-})=U_{k,\pi}(\nu), \quad\forall k\neq j.
\end{split}
\end{equation}
Hence, we have 
\begin{equation}
\mathbf{U}_{NP}(\nu) \geq \mathbf{U}_{\pi}(\nu).
\end{equation}

%\begin{figure}
%\includegraphics[scale=0.25]{figure/thm2pfgraph.eps}
%\raggedleft
%%\justify
%\caption{Illustration of the case when the link $(i,j)$ is sending a stale packet at time $\tau$ under policy $NP$. The packet $h$ arrives to node $i$ at time $a\leq a_\pi$ under policy $NP$.}\label{Fig:illust}
%\vspace{-0.3cm}
%\end{figure}

Case 2: If $s_{NP}<s_\pi$. Let $a_\pi$ represent the arrival time of packet $s_\pi$ to node $i$ under policy $\pi$. The transmission starting time of the delivered packets over the link $(i,j)$ is denoted by $\tau$ under both policies. Apparently, $a_\pi\leq\tau\leq\nu^{-}$. Since  packet $s_\pi$ arrived to node $i$  at time $a_\pi$ in policy $\pi$, we get
\begin{equation}\label{pf23}
s_\pi\leq U_{i,\pi}(a_\pi).
\end{equation}
From \eqref{hyp1np}, we obtain 
\begin{equation}\label{pf24}
U_{i,\pi}(a_\pi)\leq U_{i,NP}(a_\pi).
\end{equation}
Combining \eqref{pf23} and \eqref{pf24}, yields
\begin{equation}\label{pf25}
s_\pi\leq U_{i,NP}(a_\pi).
\end{equation}
Hence, in policy $NP$, node $i$ has a packet with generation time no smaller than $s_\pi$ by the time $a_\pi$. Because the $U_{i,NP}(t)$ is a non-decreasing function of $t$ and $a_\pi\leq \tau$, we have
\begin{equation}\label{pf26}
U_{i,NP}(a_\pi)\leq U_{i,NP}(\tau).
\end{equation}
Then, \eqref{pf25} and \eqref{pf26} imply
\begin{equation}\label{pf26'}
s_\pi\leq U_{i,NP}(\tau).
\end{equation}
Since $s_{NP} < s_\pi$, \eqref{pf26'} tells us
\begin{equation}\label{pf26''}
s_{NP} < U_{i,NP}(\tau),
\end{equation}
and hence policy $NP$ is sending a stale packet on link $(i,j)$. By the definition of policy $NP$, this happens only when all packets generated after $s_{NP}$ in the queue of the link $(i,j)$ have been delivered to node $j$ by time $\tau$. In addition, \eqref{pf26'} tells us that by time $\tau$, node $i$ has already received a packet (say packet $h$) generated no earlier than $s_\pi$ in policy $NP$. By $s_{NP} < s_\pi$, packet $h$ is generated after $s_{NP}$. Hence, packet $h$ must have been delivered to node $j$ by time $\tau$ in policy $NP$ such that 
%This means that by time $\tau$, node $i$ under policy $NP$ has a packet with generation time no smaller than $s_\pi$. Since $s_{NP}<s_\pi$, the policy $NP$ is not sending the freshest packet at node $i$ at time $\tau$, as shown in Fig. \ref{Fig:illust}. By the definition of policy $NP$, this happens in policy $NP$ only when all packets that are fresher than the packet $s_{NP}$ (including the packet $s_\pi$ or a fresher one) in the queue of link $(i,j)$ have already been delivered to node $j$, and policy $NP$ is sending a stale packet in the queue of link $(i,j)$. This means that a packet with generation time no smaller than $s_\pi$ must have arrived to node $j$ under policy $NP$ before time $\tau$. thus, we have
\begin{equation}\label{pf27}
s_\pi\leq U_{j,NP}(\tau).
\end{equation}
Because the $U_{j,NP}(t)$ is a non-decreasing function of $t$, and $\tau \leq \nu^{-}$, \eqref{pf27} implies
%\begin{equation}\label{pf28}
%U_{j,NP}(\tau)\leq U_{j,NP}(\nu^{-}).
%\end{equation}
%Combining \eqref{pf27} and \eqref{pf28}, we get
\begin{equation}\label{pf29}
s_\pi\leq U_{j,NP}(\nu^{-}).
\end{equation}
Also, from \eqref{hyp1np}, we have
\begin{equation}\label{pf230}
U_{j,\pi}(\nu^{-})\leq U_{j,NP}(\nu^{-}).
\end{equation}
Combining \eqref{pf29} and \eqref{pf230} with \eqref{Def1np}, we obtain
\begin{equation}
U_{j,NP}(\nu)\geq U_{j,\pi}(\nu).
\end{equation}
Since there is no packet delivery under other links, we get
\begin{equation}
\begin{split}
U_{k,NP}(\nu)&=U_{k,NP}(\nu^{-})\\&\geq U_{k,\pi}(\nu^{-})=U_{k,\pi}(\nu), \quad \forall k\neq j.
\end{split}
\end{equation}
Hence, we have 
\begin{equation}
\mathbf{U}_{NP}(\nu) \geq \mathbf{U}_{\pi}(\nu),
\end{equation}
which complete the proof.
\end{proof}

\begin{lemma}\label{lem4np}
Suppose that under policy $NP$, $\mathbf{U'}_{NP}$ is obtained by the arrival of a new packet to node $0$ in the system whose state is $\mathbf{U}_{NP}$. Further, suppose that under policy $\pi$, $\mathbf{U'}_{\pi}$ is obtained by the arrival of a new packet to node $0$ in the system whose state is $\mathbf{U}_\pi$. If
\begin{equation}\label{hyp2np}
 \mathbf{U}_{NP} \geq \mathbf{U}_\pi,
\end{equation}
then,
\begin{equation}
\mathbf{U'}_{NP} \geq \mathbf{U'}_{\pi}.
\end{equation}
\end{lemma}

\begin{proof}
The proof of Lemma \ref{lem4np} is similar to that of Lemma \ref{lem4}, and hence is not provided.
%Let $s$ denote the generation time of the new arrived packet. From the definition of the system state, we can deduce that
%\begin{equation}\label{Def2np}
%\begin{split}
%U_{0,NP}'&=\max\{U_{0,NP},s\},\\
%U_{0,\pi}'&=\max\{U_{0,\pi},s\}.
%\end{split}
%\end{equation}
%Combining this with \eqref{hyp2np}, we obtain
%\begin{equation}
%U'_{0,NP}\geq U'_{0,\pi}.
%\end{equation}
%Since there is no packet delivery under other links, we get
%\begin{equation}
%U'_{k,NP}=U_{k,NP}\geq U_{k,\pi}=U'_{k,\pi}, \quad \forall k\neq 0.
%\end{equation}
%Hence, we have 
%\begin{equation}
%\mathbf{U'}_{NP} \geq \mathbf{U'}_{\pi},
%\end{equation}
%which complete the proof. 
\end{proof}

\begin{proof}[ Proof of Lemma \ref{lem2np}]
For any sample path, we have that $\mathbf{U}_{NP}(0^-) = \mathbf{U}_{\pi}(0^-)$. This, together with Lemma \ref{lem3np} and Lemma \ref{lem4np},  implies that  
\begin{equation}
\begin{split}
[\mathbf{U}_{NP}(t)\vert\mathcal{I}] \geq [\mathbf{U}_{\pi}(t)\vert\mathcal{I}],\nonumber
\end{split}
\end{equation}
holds for all $t\in[0,\infty)$. Hence, \eqref{main_eqnp} holds which implies \eqref{law9np} by Theorem 6.B.30 in \cite{shaked2007stochastic}.
This completes the proof.
\end{proof}

\begin{proof}[Proof of Theorem \ref{thm2}]
As a result of Lemma \ref{lem2np}, we have
\begin{equation*}
\begin{split}
[\{\mathbf{U}_{NP}(t),  t\in[0,\infty)\}\vert\mathcal{I}]\geq_{\text{st}} [\{\mathbf{U}_{\pi}(t), t\in[0,\infty)\}\vert\mathcal{I}],
 \end{split}
\end{equation*}
holds for all $\pi\in\Pi_{npwc}$, which implies
\begin{equation*}
\begin{split}
[\{\mathbf{\Delta}_{NP}(t), t\in[0,\infty)\}\vert\mathcal{I}]\!\!\leq_{\text{st}} \!\![\{\mathbf{\Delta}_{\pi}(t), t\in[0,\infty)\}\vert\mathcal{I}],
 \end{split}
\end{equation*}
holds for all $\pi\in\Pi_{npwc}$. This completes the proof. 
\end{proof}
\fi

\bibliographystyle{IEEEbib}
\bibliography{MyLib}
\end{document}